\documentclass[notitlepage,aps,prd,twocolumn,superscriptaddress,floatfix, nofootinbib]{revtex4-1}

\usepackage[utf8]{inputenc}

\usepackage{enumitem}
\usepackage{braket}
\usepackage{hyperref}
\usepackage{comment}
\usepackage{amsmath,amssymb,amsfonts,mathrsfs}
\usepackage{booktabs}
\usepackage{mathtools}
\usepackage{bm}
\usepackage{dcolumn}
\usepackage{graphicx}   
\usepackage{latexsym}
\usepackage[dvipsnames]{xcolor}
\usepackage{colortbl}
\usepackage{amsthm}

\usepackage[acronym]{glossaries}

\hypersetup{
  colorlinks = true,
  linkcolor = blue,
  urlcolor = blue, 
  citecolor = blue
}

\newcommand{\dd}{\mathrm{d}}
\newcommand{\pd}{\partial}
\newcommand{\PD}[2]{\ensuremath{\frac{\partial #1}{\partial #2}}}
\newcommand{\TD}[2]{\ensuremath{\frac{\dd #1}{\dd #2}}}

\newcommand{\ud}[2]{^{#1}{}_{#2}}
\newcommand{\du}[2]{_{#1}{}^{#2}}

\newcommand{\dud}[3]{_{#1}{}^{#2}{}_{#3}}
\DeclareMathOperator{\e}{e}

\newcommand{\M}{\mathcal{M}}
\newcommand{\F}{\mathcal{F}}
\newcommand{\X}{\mathfrak{X}}

\newcommand{\ce}{\coloneqq}
\newcommand{\ec}{\eqqcolon}

\newcommand{\bbR}{\mathbb{R}}

\begin{document}

\newtheorem{definition}{Definition}[section]
\newtheorem{example}{Example}[section]
\newtheorem{theorem}{Theorem}[section]
\newtheorem{corollary}{Corollary}[section]
\newtheorem{lemma}{Lemma}[section]
\newtheorem{remark}{Remark}[section]
\newtheorem{proposition}{Proposition}[section]

\newpage

\title{Autoparallels and the Inverse Problem of the Calculus of Variations}

\author{Lavinia Heisenberg}
\email{heisenberg@thphys.uni-heidelberg.de}
\affiliation{Institute for Theoretical Physics, University of Heidelberg, Philosophenweg 16
D-69120	Heidelberg
Germany}

\date{\today}


\begin{abstract}
We prove that autoparallel curves associated with a torsion-free but not necessarily metric-compatible affine connection can be derived from an action principle. We explicitly construct the action functional and show by standard variational techniques that it produces the desired equations. Our analysis is based on systematically solving the inverse problem of the calculus of variation and the associated Helmholtz conditions. This demonstrates that the dynamics of autoparallels admit a consistent variational formulation even in the presence of non-metricity. Our results provide a variational framework for particle motion in metric-affine geometries and thereby contribute to the mathematical foundations of the geodesic principle in relativistic gravity.
\end{abstract}



\maketitle


\section{Introduction}\label{sec:intro}
Einstein's original formulation of General Relativity (GR) builds on the assumption of a torsion-free and metric-compatible affine structure. More general geometric frameworks, and with them a rich landscape of gravity theories, naturally arise once one allows the affine structure of spacetime to be independent of the metric~\cite{Heisenberg:2018vsk, Heisenberg:2023lru}. In a metric-affine geometry, spacetime is equipped with both a metric tensor and an affine connection, which need not be compatible. The failure of metric compatibility is measured by the non-metricity tensor, while the antisymmetric part of the connection defines the torsion. In particular, GR can be formulated equivalently in terms of curvature, torsion, or non-metricity, a correspondence often referred to as the geometric trinity of gravity~\cite{Heisenberg:2018vsk, BeltranJimenez:2019esp, Heisenberg:2023lru}.

The three theories making up the geometric trinity, namely GR, the teleparallel equivalent of GR (TEGR), and the symmetric teleparallel equivalent of GR (STEGR), are based on different geometric postulates, but their field equations possess the same solution space. This is the sense in which they are equivalent. However, conceptual differences persist and carry over into other theories of gravity based on metric-affine geometry. In particular, the motion of test particles may be affected by the particular type of geometry used to formulate a given theory of gravity.

In GR, test particles are postulated to follow geodesics. Geometrically, these are curves which extremize the proper time functional and may be interpreted as locally shortest or extremal paths. Equivalently, they satisfy the condition that their tangent vectors are parallel transported with respect to the Levi-Civita connection. The coincidence of these two characterizations of geodesics---the variational and the affine---relies crucially on the metric compatibility and torsion-free nature of the Levi-Civita connection.

In the presence of an independent affine connection, the description of free-fall motion becomes more subtle. Two a priori distinct notions of preferred curves arise. On the one hand, autoparallels are defined as curves whose tangent vectors are parallel transported with respect to the affine connection; they represent the straightest possible curves associated with the connection. On the other hand, geodesics are obtained as extremals of a point-particle action which only depends on the metric. While in (pseudo-) Riemannian geometry these two notions coincide, as mentioned above, in a generic metric-affine geometry they constitute two genuinely distinct concepts. This leads to an apparent ambiguity in the kinematics of freely falling particles.

This ambiguity raises a fundamental question at the interface of geometry and variational principles: to what extent can autoparallel curves of a general affine connection be obtained as extrema of a variational principle? In this letter we address the slightly more restricted question of whether a reparametrization-invariant scalar worldline action exists,  which only depends on the embedding of the particle worldline into spacetime and whose Euler-Lagrange equations reproduce the autoparallel equations of a torsion-free affine connection with arbitrary non-metricity. We provide a definitive answer in the affirmative to this question. To our knowledge, this is the first time an action functional generating autoparallels of a general torsion-free affine connection has been constructed.

 We achieve this by examining and solving the inverse problem of the calculus of variation. In particular, we show that, in the absence of torsion, the Helmholtz conditions can always be satisfied for a given metric and a given non-metricity tensor. We provide an explicit expression for the action functional and prove that it indeed reproduces the sought-after autoparallel equations by means of standard variational techniques. 

The letter is organized as follows. In Section~\ref{sec:GeometricPreliminaries} we review the necessary elements of metric-affine geometry to make the exposition self-contained. In Section~\ref{sec:InverseProblem} we formulate the inverse problem of the calculus of variation, adapted to our equations of interest---namely the autoparallel equation. We subsequently show how to solve the Helmholtz conditions, which culminates in a proof of existence of an action functional for the autoparallel equations in absence of torsion.  We then explicitly state the action functional and prove that it produces the desired equations. We conclude with a summary and outlook in Section~\ref{sec:discussion}.


\section{Geometric Preliminaries}\label{sec:GeometricPreliminaries}
Tensor fields evaluated at different points of a manifold cannot be directly compared. The reason is that these fields are defined over different tangent spaces, which cannot be canonically identified with each other. However, these tangent spaces are isomorphic to each other, which allows us to introduce the concept of an affine connection $\nabla$ and the related concept of parallel transport of tensor fields. With these concepts in hand, we can transport a given field from one point $p$ of the manifold along a curve $\gamma$ to a different point $q$ in order to perform a comparison. In the case of a vector field, we can carry out this process such that the vector remains parallel to itself in an appropriate sense.

Throughout this letter, $\M$ is assumed to be a smooth, connected, $n$-dimensional manifold. The set of all smooth functions from $\M$ to $\bbR$ is denoted by $\F(\M)$, $\X(\M)$ is the set of all smooth vector fields on $\M$, and $\Omega(\M)$ denotes the set of all smooth $1$-forms. We now introduce an affine connection.

\begin{definition}[Affine connection]\label{def:connection}
An \emph{affine connection} on $\mathcal M$ is a map
\[
\nabla : \X(\M) \times \X(\mathcal M) \to \X(\M)\,
\qquad
(X,Y) \mapsto \nabla_X Y,
\]
which is $\F(\M)$-linear in the first argument, $\bbR$-linear in the second one, and it satisfies the Leibniz rule for the second argument.
\end{definition}
The definition of an affine connection is manifestly coordinate-independent. For practical purposes, it is often desirable to work in local coordinates. This is achieved point-wise by evaluating $\nabla$ on the coordinate basis $\{\pd_a\}$ of a given tangent space and defining the connection coefficients $\Gamma\ud{a}{bc}$ as
\begin{align}
	\nabla_{\pd_b} \partial_c \ec \Gamma\ud{a}{bc} \partial_a.
\end{align}

Notice that $\nabla$ is not $\F(\M)$-bilinear and therefore not a tensor. Indeed, one can confirm by direct computation that $\Gamma\ud{a}{bc}$ does \emph{not} transform in a tensorial way.

Once a connection has been introduced, one can define the notion of parallel transport.

\begin{definition}[Parallel transport]\label{def:paralleltransport}
Let $\gamma : I \subset \mathbb R \to \M$ be a smooth curve.
A vector field $X\in\X(\M)$ along $\gamma$ is said to be \emph{parallel transported along $\gamma$}
if
\[
\nabla_{\dot\gamma} X = 0 ,
\]
where $\dot\gamma$ denotes the tangent vector field to $\gamma$.
\end{definition}

Parallel transport depends on both, the connection and the curve. In general, parallel transport between two points of $\mathcal M$ is path dependent.

Autoparallel curves on $\M$ serve as counterparts to straight lines in flat space. Physically, one can interpret these curves as delineating the paths of particles which are uninfluenced by external forces. Mathematically, a straight line is a trajectory that preserves its tangent vector under parallel transport. Thus, an autoparallel curve is naturally defined as follows.

\begin{definition}[Autoparallel curves]\label{def:autoparallel}
Let $\nabla$ be an affine connection on $\M$.
A smooth curve $\gamma : I\subset\bbR \to \M$ is called an \emph{autoparallel}
if its tangent vector is parallel transported along itself, i.e.
\[
\nabla_{\dot\gamma} \dot\gamma = 0 .
\]
\end{definition}

In a local chart, and using an affine parametrization, the autoparallel equation takes the form
\begin{equation}\label{eq:autoparallel}
	\TD{^2 \gamma^{a}}{\lambda^2} + \Gamma\ud{a}{bc}\TD{\gamma^b}{\lambda}\TD{\gamma^c}{\lambda}= 0 ,
\end{equation}
where $\lambda$ is an affine parameter along the curve.

Apart from the general affine connection of Definition~\ref{def:connection} we assume that the manifold is endowed with an independent metric. In a general metric-affine geometry, the connection
and the metric are completely independent from each other. 
\begin{definition}[Metric-affine geometry]
	A \emph{metric-affine geometry} is a triple $(\M, g, \nabla)$, where $g$ is a smooth, non-degenerate symmetric $(0,2)$-tensor field and $\nabla$ an affine connection.
\end{definition}
The signature of the metric has no influence on the results of this letter. We therefore keep it unspecified.

Any metric-affine geometry can be characterized by three independent tensors, which are constructed from the affine connection and the metric. The simplest tensor to consider is the torsion tensor.

\begin{definition}[Torsion tensor]
	The \emph{torsion tensor} is an $\F(\M)$-trilinear map
	\[
	T: \Omega^1(\M)\times\X(\M)\times\X(\M)\to \F(\M)
	\]
	defined by
	\[
	T(\omega, X, Y) = \omega(\nabla_X Y) - \omega(\nabla_Y X) - \omega([X,Y])
	\]
	for any $1$-form $\omega\in\Omega(\M)$ and vector fields $X,Y\in\X(\M)$, and where $[X,Y]$ denotes the Lie bracket.
\end{definition}
It follows directly from the definition that the torsion tensor is independent of the metric and that it is anti-symmetric in the last two slots:
\begin{align}
	T(\omega, X, Y) = - T(\omega, Y, X).
\end{align}
In a local coordinate system, the components of the torsion tensor are given by
\begin{align}
	T\ud{a}{bc} = 2\Gamma\ud{a}{[bc]}.
\end{align}
Next, we introduce the curvature tensor, which is also independent of the metric.
\begin{definition}[Curvature tensor]
	The \emph{curvature tensor} is an $\F(\M)$-quadrilinear map
	\[
	R:\X(\M)\times\X(\M)\times \X(\M)\times\Omega^1(\M) \to \F(\M)
	\] 
	defined by
	\[
		R(X,Y,Z,\omega) = \omega(\nabla_X\nabla_Y Z) - \omega(\nabla_Y \nabla_X Z) - \omega(\nabla_{[X,Y]}Z)
	\] 
	for any $1$-form $\omega\in\Omega(\M)$ and vector fields $X,Y,Z\in \X(\M)$. 	
\end{definition}
This definition directly implies that $R$ is anti-symmetric in the first two slots:
\begin{align}
	R(X,Y,Z,\omega) = - R(Y,X,Z,\omega).
\end{align}
In a local coordinate system, the components of the curvature tensor take the form
\begin{align}
	R\du{abc}{d} = 2\pd_{[b} \Gamma\ud{d}{a]c} + 2\Gamma\ud{d}{[b|i}\Gamma\ud{i}{a]c}.
\end{align}
Finally, we introduce the non-metricity tensor, which depends on the affine connection as well as on the metric tensor.
\begin{definition}(Non-metricity tensor)
	The non-metricity tensor is an $\F(\M)$-trilinear map
	\[
	Q:\X(\M) \times \X(\M)\times \X(\M)\to \F(\M)
	\]
	defined by
	\begin{align*}
		Q(X,Y,Z) &= (\nabla_X g)(Y,Z) \\
		&= X(g(Y,Z)) - g(\nabla_X Y, Z) - g(Y,\nabla_X Z)
	\end{align*}
	for any vector fields $X,Y,Z\in\X(\M)$.
\end{definition}
This definition implies that the non-metricity tensor is symmetric in the last two slots:
\begin{align}
	Q(X,Y,Z) = Q(X,Z,Y).
\end{align}
Furthermore, its components in a local coordinate chart are given by
\begin{align}
	Q_{abc} = \nabla_a g_{bc} = \pd_a g_{bc} - 2\Gamma\ud{d}{a(b}g_{c)d}.
\end{align}
In pseudo-Riemannian geometry, neither torsion nor non-metricity are present. This special type of metric-affine geometry serves as a reference model in what follows and we recall some basic facts, beginning with a definition.
\begin{definition}[Pseudo-Riemannian manifold]
	A \emph{pseudo-Riemannian manifold} is a pair $(\M, g)$, where $g$ is a smooth metric tensor and the affine connection is determined by vanishing torsion and vanishing non-metricity.
\end{definition}
As is well-known, there is a \emph{unique} connection for which torsion and non-metricity vanish. This is the \emph{Levi-Civita connection}, whose coefficients with respect to a local coordinate chart we denote by $\left\{a\atop bc\right\}$.  These Christoffel symbols are completely determined by the metric:
\begin{align}
	\left\{a \atop bc\right\} = \frac12 g^{ad}\left(\pd_b g_{cd} + \pd_c g_{bd} - \pd_d g_{bc}\right).
\end{align}
In a generic metric-affine geometry $(\M, g, \nabla)$, we can decompose the general connection $\nabla$ with respect to torsion, non-metricity, and the Christoffel symbols of the Levi-Civita connection:

\begin{proposition}[Decomposition of the connection]\label{pro:disCon}
Let $(\mathcal M,g,\nabla)$ be a metric-affine manifold. Then the connection admits the decomposition
\[
\Gamma\ud{a}{bc} = \left\{a\atop bc\right\} + L\ud{a}{bc} + K\ud{a}{bc},
\]
where $L\ud{a}{bc}$ and $K\ud{a}{bc}$ are the components of the \emph{disformation} and \emph{contortion} tensors, respectively, which are defined as
\begin{align*}
	L\ud{a}{bc} &\ce \frac12 Q\ud{a}{bc} - Q\dud{(b}{a}{c)} \\
	K\ud{a}{bc} &\ce \frac12 T\ud{a}{bc} + T\dud{(b}{a}{c)}.
\end{align*}
\end{proposition}
This proposition makes clear one sense in which a pseudo-Riemannian geometry acts as a reference model for more general metric-affine geometries, since the Christoffel symbols of the Levi-Civita connection inevitably appear in the local expression for~$\Gamma\ud{a}{bc}$.

Using the above decomposition, we can express the autoparallel equation in a suggestive manner:

\begin{corollary}[Autoparallel equation, local form]
In a local chart, and using the decomposition of Proposition~\ref{pro:disCon}, the affinely parametrized autoparallel equation~\eqref{eq:autoparallel} takes the form
\begin{align}\label{eq:autoparallelL}
	\TD{^2\gamma^{a}}{\lambda^2} + \left\{a\atop bc\right\}\TD{\gamma^b}{\lambda} \TD{\gamma^c}{\lambda} = -\left(L\ud{a}{bc} + K\ud{a}{bc}\right)\TD{\gamma^b}{\lambda} \TD{\gamma^c}{\lambda}.
\end{align}
\end{corollary}
The left hand side is recognized to be equal to the left hand side of the geodesic equation, while the right hand side acquires the interpretation of a force which induces a deviation from the geodesic path. 

We recall that in a pseudo-Riemannian geometry, curves of extremal length between two points can be found by means of a variational principle.
\begin{definition}[Geodesics]\label{def:metricgeodesic}
Let $(\mathcal M,g)$ be a pseudo-Riemannian manifold.
A smooth curve $\gamma : I \to \mathcal M$ is called a \emph{geodesic} if it is an extremal
of the functional
\[
	\mathcal{S}[\gamma] = \int_I \sqrt{\lvert g_{ab}\dot{\gamma}^{a} \dot{\gamma}^{b}\rvert}\, \dd\lambda .
\]
\end{definition}
Varying this functional, while holding the boundary points fixed, results in the geodesic equation.

\begin{proposition}[Geodesic equation]\label{pro:geodesics}
Geodesics can be re-parametrized such that they satisfy the equation
	\begin{equation}\label{eq:geodesic}
		\TD{\gamma^{a}}{\lambda^2} + \left\{ a \atop bc \right\} \TD{\gamma^b}{\lambda}\TD{\gamma^c}{\lambda} = 0 ,
	\end{equation}
where $\left\{a\atop bc\right\}$ are the Christoffel symbols of the Levi-Civita connection.
\end{proposition}
A comparison between the geodesic equation~\eqref{eq:geodesic} and the autoparallel equation~\eqref{eq:autoparallelL} shows that they are the same if and only if the geometry is pseudo-Riemannian. On a more general metric-affine manifold, however, the geodesics and autoparallel curves are genuinely distinct concepts. 

In gravitational theories based on torsion or non-metricity~\cite{BeltranJimenez:2017tkd, BeltranJimenez:2018vdo,BeltranJimenez:2019esp, Heisenberg:2023lru}, the existence of a more general affine connection thus introduces the question which paths test particle follow: Geodesics or autoparallel curves? On purely theoretical grounds, postulating that test particles follow geodesic curves might seem more appealing, because these curves are derivable from an action principle. In this letter we show that the autoparallel equation~\eqref{eq:autoparallelL} of any torsion-free geometry is determined by an action principle similar to the one in Definition~\ref{def:metricgeodesic}. 

\section{Autoparallels and the Inverse Problem of the Calculus of Variations}\label{sec:InverseProblem}

In this section we formulate the variational problem for autoparallel equations on manifolds equipped with a metric and an independent affine connection which is torsion-free but not metric-compatible. We show that the autoparallel equation~\eqref{eq:autoparallelL} can be derived from an action principle. 

\subsection{Existence of an action functional}
\begin{definition}[Scalar worldline action]
	A \emph{scalar worldline action} is a functional of the form
	\[
	\mathcal S[\gamma] = \int_{I} \mathcal{L}\bigl(\gamma(\lambda), \dot\gamma(\lambda)\bigr)\, \dd\lambda ,
	\]
	where $\mathcal{L} : T\mathcal M \to \mathbb R$ is a smooth scalar function on the tangent bundle, and $\gamma : I \to \mathcal M$ is a smooth curve with fixed endpoints.
\end{definition}

We assume throughout that the action is invariant under orientation-preserving reparametrizations of $\lambda$, unless stated otherwise.

\begin{definition}[Local worldline action]
	A scalar worldline action is said to be \emph{local} if the Lagrangian $\mathcal{L}$
	depends only on $\gamma(\lambda)$ and $\dot\gamma(\lambda)$ at the same parameter value.
	It is said to be \emph{nonlocal} if it depends on integrals along the curve.
\end{definition}

In GR, the action functional which produces the geodesic equation is interpreted as the proper time of a test particle. In metric-affine geometries with vanishing torsion, a generalized proper time can be introduced~\cite{Romero_2012,Avalos17, DelhomTime}, which takes the form
\begin{align}\label{eq:ModifiedProperTime}
	\mathcal{S}[\gamma] = \int_I \e^{-\frac{\omega(\lambda)}{2}}\sqrt{|g_{ab}\dot{\gamma}^{a} \dot{\gamma}^b|}\,\dd\lambda
\end{align}
with
\begin{align}\label{eq:DefOmega}
	\omega(\lambda) \ce \int_{\lambda_1}^\lambda \frac{Q_{abc} \dot{\gamma}^{a}\dot{\gamma}^b\dot{\gamma}^c}{g_{ab}\dot{\gamma}^{a} \dot{\gamma}^b}\dd \lambda'\,.
\end{align}
This is an example of a \emph{non-local} action functional which, more importantly, \emph{fails} to re-produce the autoparallel equation~\eqref{eq:autoparallelL}. Only when the integrand of $\omega$ is an exact $1$-form does one obtain a local action. Moreover, in this special case the Euler-Lagrange equations can be interpreted as autoparallel equations with a non-metricity tensor of Weyl form:
\begin{align}
	Q_{abc} = \partial_a \omega\, g_{bc}.
\end{align}
For later reference, we provide a precise definition for Weyl-type connections: 
\begin{definition}[Weyl-type connection]
A torsion-free affine connection $\nabla$ on $(\mathcal M,g)$ is said to be of
\emph{Weyl-type} if there exists a smooth $1$-form $\omega \in \Omega(\M)$ such that the non-metricity tensor takes the form
\[
Q_{abc} = \omega_a g_{bc}.
\]
\end{definition}

Rather than seeking modifications of the generalized proper time functional~\eqref{eq:ModifiedProperTime} or otherwise trying to guess the correct functional for reproducing the autoparallel equation for a \emph{generic} choice of non-metricity, our approach is to systematically solve the inverse problem of the calculus of variation. 

The fundamental question of the inverse problem of the calculus of variation can be formulated as follows: Let $\gamma:I\subset \bbR\to\M$ be a smooth curve subjected to the equation
\begin{align}\label{eq:EOMGeneral}
	\ddot{\gamma}^{a} - F^{a}(\gamma, \dot{\gamma}) = 0
\end{align}
for some ``force'' $F^{a}$. Under which conditions is the existence of a symmetric, non-degenerate tensor $H_{ab} = H_{ab}(\gamma, \dot{\gamma})$ and a Lagrangian $\mathcal{L}=\mathcal{L}(\gamma, \dot{\gamma})$ guaranteed, such that
\begin{align}
	H_{ab} \left(\ddot{\gamma}^b - F^b\right) = \frac{\dd}{\dd \lambda}\PD{\mathcal{L}}{\dot{\gamma}^{a}} - \PD{\mathcal{L}}{\gamma^{a}}
\end{align}
is true?
The Sonin-Douglas theorem~\cite{InverseProblemBook} provides an answer to this question.
\begin{theorem}[Sonin-Douglas]
	The following two statements are equivalent:
	\begin{itemize}
		\item[a)] The system~\eqref{eq:EOMGeneral} is variational (i.e., it can be derived from a Lagrangian and its associated Euler-Lagrange equations);
		\item[b)] There exists a symmetric, non-degenerate tensor $H_{ab}$ which, together with the force $F^{a}$, satisfy the Helmholtz conditions
		\begin{itemize}
			\item[(H1)] $\PD{H_{ab}}{\dot{\gamma}^c} - \PD{H_{cb}}{\dot{\gamma}^{a}} = 0$
			\item[(H2)] $\frac{\dd H_{ab}}{\dd \lambda} + \frac12 \left(H_{ac}\PD{F^c}{\dot{\gamma}^b} + H_{bc} \PD{F^c}{\dot{\gamma}^{a}}\right) = 0$ 
			\item[(H3)] $\left(\PD{H_{ab}}{\gamma^c} - \PD{H_{cb}}{\gamma^{a}}\right) F^{b} + H_{ab} \PD{F^{b}}{\gamma^c} - H_{cb} \PD{F^{b}}{\gamma^{a}} - \frac12 \dot{\gamma}^b\PD{}{\gamma^{b}}\left(H_{ad}\PD{F^d}{\dot{\gamma}^c} - H_{cd}\PD{F^d}{\dot{\gamma}^{a}}\right)$.
		\end{itemize}
	\end{itemize}
\end{theorem}
The Helmholtz conditions can be read in two distinct ways: Given $H_{ab}$, what is the most general force $F^{a}$ consistent with $\emph{(H1)-(H3)}$? Or: Given a force $F^{a}$, what is the most general $H_{ab}$ consistent with the Helmholtz conditions?

We offer two examples for the first scenario.

\begin{proposition}
	Let $g_{ab}$ be a metric on $\M$ (the signature is irrelevant for what follows) and set $H_{ab} = g_{ab}$. The Helmholtz conditions then imply that the most general force consistent with this choice of $H_{ab}$ is
	\begin{align}
		F^{a} = -\left\{a\atop bc\right\}\dot{\gamma}^b \dot{\gamma}^c + P\ud{a}{b}\dot{\gamma}^b + S^{a}
	\end{align}
	for some tensors $P\ud{a}{b}(\gamma)$ and $S^{a}(\gamma)$, and where $\left\{a\atop bc\right\}$ are the Christoffel symbols of the Levi-Civita connection.
\end{proposition}
Clearly, this force gives rise to the geodesic equation as a special case. It also contains the Lorentz equation, which describes the motion of a charged test particle in the presence of an electromagnetic and gravitational field, as a special case. A proof of this proposition can be found in~\cite{InverseProblemBook}.

By a minimal and straightforward adaptation of the technique used in~\cite{InverseProblemBook}, one can also prove the following 
\begin{proposition}
	Let $g_{ab}$ be a metric on $\M$ and $\omega\in \F(\M)$. Set $H_{ab} = \e^{-\omega} g_{ab}$. Then the Helmholtz conditions imply that the most general force consistent with this choice of $H_{ab}$ is
	\begin{align}
		F^{a} = -\left(\left\{a\atop bc\right\} + L\ud{a}{bc}\right)\dot{\gamma}^b \dot{\gamma}^c + P\ud{a}{b}\dot{\gamma}^b + S^{a}
	\end{align}
	for some tensors $P\ud{a}{b}(\gamma)$ and $S^{a}(\gamma)$, and where the disformation tensor $L\ud{a}{bc}$ is constructed from the Weyl non-metricity tensor $Q_{abc} = \partial_a\omega \,g_{bc}$. Thus, it is given by
	\begin{align}
		L\ud{a}{bc} = \frac12 g^{ad}\pd_d \omega\, g_{bc} - \frac12 \left(\pd_b \omega \,\delta\ud{a}{c} + \pd_c\omega \, \delta\ud{a}{b}\right).
	\end{align}
\end{proposition}
The action which produces the first term, proportional to $\dot{\gamma}^b \dot{\gamma}^c$, is precisely~\eqref{eq:ModifiedProperTime} with $\omega$ taken to be an arbitrary smooth function which only depends on $\gamma$, rather than being defined through~\eqref{eq:DefOmega}.

Now we turn our attention to the second scenario. Our goal is to find the most general $H_{ab}$ for the autoparallel equation in a metric-affine geometry $(\M, g, \nabla)$, where $\nabla$ is taken to be torsion-free but otherwise not further specified. Let $\Gamma\ud{a}{bc} = \left\{a\atop bc\right\} + L\ud{a}{bc}$ be the coefficients of this affine connection. The force term is then given by
\begin{align}
	F^{a} = -\Gamma\ud{a}{bc}\dot{\gamma}^b\dot{\gamma}^c
\end{align}
and the Helmholtz conditions can be re-written as
\begin{itemize}
	\item[(H1)] $\PD{H_{ab}}{\dot{\gamma}^c} - \PD{H_{cb}}{\dot{\gamma}^{a}} = 0$
	\item[(H2)] $\dot{\gamma}^{a}\nabla_a H_{bc} - \dot{\gamma}^{a} \dot{\gamma}^{e} \Gamma\ud{d}{ae} \PD{H_{bc}}{\dot{\gamma}^d} = 0$
	\item[(H3)] $\dot{\gamma}^{a} \dot{\gamma}^{b}( H_{ic} R\du{kab}{c} - H_{kc}R\du{iab}{c} + \Gamma\ud{c}{ka}\nabla_b H_{ic} - \Gamma\ud{c}{ia}\nabla_b H_{kc} + \Gamma\ud{c}{ab}\nabla_i H_{kc} - \Gamma\ud{c}{ab}\nabla_k H_{ic}) = 0$ 
\end{itemize}
where we made use of the affine connection $\nabla$ and its associated curvature tensor $R$. Condition (H2) can be read as a first order partial differential equation for $H_{ab}$. However, the second term in this equation is not tensorial. For this equation to be consistent with the tensorial nature of $H_{ab}$, we have to demand that the second term vanishes for every possible curve $\gamma$. For a non-trivial $\Gamma\ud{a}{bc}$ this is only possible if
\begin{align}
	\PD{H_{ab}}{\dot{\gamma}^{c}} = 0.
\end{align}
Thus, $H_{ab}$ can only depend on $\gamma$, which implies that condition (H1) is trivially satisfied. Moreover, because $\dot{\gamma}^{a} \nabla_a H_{bc} = 0$ has to hold for any curve $\gamma$, we can conclude that 
\begin{align}\label{eq:MetricCompatibility}
	\nabla_a H_{bc} = 0.
\end{align}
We recall that $H_{ab}$ is symmetric and non-degenerate. This suggests that we can regard $H_{ab}$ as a second metric on $\M$ (albeit a metric which is \emph{not} used to raise or lower indices) and~\eqref{eq:MetricCompatibility} can be read as the condition of metric compatibility.

Because $\nabla$ is torsion-free, and because $H_{ab}$ is symmetric and non-degenerate, we can solve~\eqref{eq:MetricCompatibility} for the connection coefficients in the usual manner:
\begin{align}\label{eq:ChristoffelSymbolsH}
	\Gamma\ud{a}{bc} = \frac12 (H^{-1})^{ad}\left(\partial_b H_{cd} + \partial_c H_{bd} - \partial_d H_{bc}\right).
\end{align}
Here, $(H^{-1})^{ab}$ is the inverse of $H_{ab}$, which is guaranteed to exists due to $H$'s non-degeneracy. Importantly, $(H^{-1})^{ab}$ is \emph{not} obtained from $H_{ab}$ by raising its indices with the metric $g_{ab}$. The right hand side of~\eqref{eq:ChristoffelSymbolsH} can be regarded as the Christoffel symbols of $H_{ab}$ and the whole equation tells us that $H_{ab}$ encodes information about the metric $g_{ab}$, its Christoffel symbols, and the non-metricity tensor $Q_{abc}$.

What remains to be verified is condition (H3). Upon using~\eqref{eq:MetricCompatibility}, it greatly simplifies to
\begin{align}\label{eq:H3Simplified}
	\dot{\gamma}^{a} \dot{\gamma}^{b}( H_{ic} R\du{kab}{c} - H_{kc}R\du{iab}{c} ) = 0.
\end{align}
By defining $\bar{R}_{abcd} \ce R\du{abc}{k}H_{kd}$, we can rewrite this condition equivalently as
\begin{align}\label{eq:H3SimplifiedNewVariables}
	\dot{\gamma}^{a} \dot{\gamma}^{b}(\bar{R}_{kabi} - \bar{R}_{iabk}) = 0.
\end{align}
Using~\eqref{eq:ChristoffelSymbolsH}, one can verify directly that $\bar{R}_{abcd}$ has the same symmetries as the Riemann tensor\footnote{The Riemann tensor is the curvature tensor of the Levi-Civita connection. This tensor has more symmetries than the curvature tensor of a generic affine connection.}:
\begin{align}\label{eq:SymmetriesRbar}
	\bar{R}_{abcd} &= - \bar{R}_{bacd} \notag\\
	\bar{R}_{abcd} &= - \bar{R}_{abdc} \notag\\
	\bar{R}_{abcd} &=  \phantom{-}\bar{R}_{cdab}
\end{align}
This can also be seen by noticing that $R\du{abc}{d}$ is the curvature tensor with respect to the Levi-Civita symbols~\eqref{eq:ChristoffelSymbolsH} of the effective metric $H_{ab}$.

Using these symmetries, it follows that~\eqref{eq:H3SimplifiedNewVariables}, and therefore~\eqref{eq:H3Simplified}, is satisfied without imposing any further restrictions on $H_{ab}$ or $R\du{abc}{d}$. 

In summary, we found that the autoparallel equation is derivable from a Lagrangian if there exists a tensor $H_{ab}$ with the following four properties: It is symmetric, non-degenerate, it is independent of $\dot{\gamma}^{a}$, and it satisfies~\eqref{eq:MetricCompatibility}. If such a tensor exists, it follows directly from these properties, together with the torsion-freeness of $\nabla$, that $H_{ab}$ encodes information about the metric and the non-metricity tensor. This is the content of equation~\eqref{eq:ChristoffelSymbolsH}.

To definitively answer the question whether an action for the autoparallel equations exists, we need to address the question whether~\eqref{eq:MetricCompatibility} admits non-trivial solutions and whether these solutions are unique, once initial conditions have been specified. Since~\eqref{eq:MetricCompatibility} constitutes a first order system of partial differential equations with more equations than unknown functions, the system is overdetermined and there have to be integrability conditions. Indeed, Frobenius-type conditions need to hold in our case:
\begin{align}
	(\nabla_a\nabla_b - \nabla_b\nabla_a)H_{cd} = 0.
\end{align}
This integrability condition can be re-written in terms of the curvature tensor $\bar{R}_{abcd}$:
\begin{align}
	\bar{R}_{abcd} + \bar{R}_{abdc} = 0.
\end{align}
Because of the anti-symmetry in the last two indices (see equation~\eqref{eq:SymmetriesRbar}), we conclude that the integrability condition is always satisfied. This ensures that locally a unique solution $H_{ab}$ exists once $H_{ab}$ has been specified in one point (initial condition).

We have therefore shown that a Lagrangian exists whose associated Euler-Lagrange equations produce the autoparallel equations for a torsion-free affine connection, provided the solution to~\eqref{eq:MetricCompatibility} satisfies $\det H_{ab} \neq 0$.

\subsection{The action and its Euler-Lagrange equations}
The four defining properties of $H_{ab}$ suggest that it should be treated as an auxiliary metric on $(\M, g, \nabla)$. Once the metric and the non-metricity are given, the associated $H_{ab}$ can be constructed.
The main result of this letter is the following
\begin{theorem}[Action of autoparallel equation]
	The Euler-Lagrange equations associated with the variation of the functional 
	\begin{align}\label{eq:ActionQ}
	\mathcal{S}[\gamma] = \int_I \sqrt{|H_{ab}\dot{\gamma}^{a} \dot{\gamma}^b|}\,\dd\lambda,
	\end{align}
	where $H_{ab}$ is symmetric and a non-degenerate solution of~\eqref{eq:MetricCompatibility}, produce the autoparallel equations~\eqref{eq:autoparallelL} for a generic disformation tensor $L\ud{a}{bc}$ and a vanishing contortion tensor $K\ud{a}{bc}$:
	\begin{align}\label{eq:autoparallelQ}
	\TD{^2\gamma^{a}}{\lambda^2} + \left\{a\atop bc\right\}\TD{\gamma^b}{\lambda} \TD{\gamma^c}{\lambda} = -L\ud{a}{bc} \TD{\gamma^b}{\lambda} \TD{\gamma^c}{\lambda}.
\end{align}
\end{theorem}\medskip
We emphasize that the above functional produces the disformation tensor associated with \emph{any} non-metricity tensor we choose, not merely the one of Weyl-type.

\begin{proof}
	Let $\mathcal{L}= \sqrt{|H_{ab}\dot{\gamma}^{a} \dot{\gamma}^b|}$. Then
	\begin{align*}
		\PD{\mathcal{L}}{\gamma^{a}} &= \frac{1}{2\mathcal{L}}\pd_a H_{bc}\dot{\gamma}^{b} \dot{\gamma}^{c} &\text{and} && \PD{\mathcal{L}}{\dot{\gamma}^{a}} &= \frac{1}{\mathcal{L}}H_{ab}\dot{\gamma}^{b}.
	\end{align*}
	The total $\lambda$-derivative of $\PD{\mathcal{L}}{\dot{\gamma}^{a}}$ is given by
	\begin{align*}
		\TD{}{\lambda}\PD{\mathcal{L}}{\dot{\gamma}^{a}} = \frac{1}{\mathcal{L}}\left(H_{ab}\ddot{\gamma}^{b} + \pd_c H_{ab}\dot{\gamma}^{c} \dot{\gamma}^{b}\right) - \frac{\dot{\mathcal{L}}}{\mathcal{L}^2} H_{ab}\dot{\gamma}^{d}.
	\end{align*}
	The Euler-Lagrange operator thus takes the form
	\begin{align*}
		\TD{}{\lambda}\PD{\mathcal{L}}{\dot{\gamma}^{a}} - \PD{\mathcal{L}}{\gamma^{a}} &= \frac{1}{\mathcal{L}}\left(H_{ab}\ddot{\gamma}^{b} + \pd_c H_{ab}\dot{\gamma}^{c} \dot{\gamma}^{b}\right) \\
		&\phantom{=} - \frac{\dot{\mathcal{L}}}{\mathcal{L}^2} H_{ab}\dot{\gamma}^{d} - \frac{1}{2\mathcal{L}}\pd_a H_{bc}\dot{\gamma}^{b} \dot{\gamma}^{c}.
	\end{align*}
	Setting to zero, multiplying by $\mathcal{L}$, contracting with $(H^{-1})^{ad}$, and simplifying yields
	\begin{align*}
		\ddot{x}^d + \frac12(H^{-1})^{ad}\left(\pd_c H_{ab} + \pd_b H_{ac} - \pd_a H_{bc}\right)\dot{\gamma}^{b} \dot{\gamma}^c = \frac{\dot{\mathcal{L}}}{\mathcal{L}}\dot{\gamma}^d.
	\end{align*}
	Using~\eqref{eq:ChristoffelSymbolsH} together with $\Gamma\ud{a}{bc} = \left\{a \atop bc\right\} + L\ud{a}{bc}$, this reduces to
	\begin{align}
		\ddot{x}^{d} + \left(\left\{d \atop bc\right\} + L\ud{d}{bc}\right)\dot{\gamma}^{b} \dot{\gamma}^c = \left(\TD{}{\lambda}\ln \mathcal{L}\right) \dot{\gamma}^d.
	\end{align}
	This is the non-affinely parametrized autoparallel equation~\eqref{eq:autoparallelL} with $K\ud{a}{bc}=0$. By choosing an affine parametrization, defined by $\dot{\mathcal{L}} = 0$, we finally obtain (after relabelling indices)
	\begin{align}\label{eq:FinalEq}
		\ddot{x}^{a} + \left(\left\{d \atop bc\right\} + L\ud{d}{bc}\right)\dot{\gamma}^{b} \dot{\gamma}^c = 0.
	\end{align}
\end{proof}

A few remarks are in order.

\begin{remark}
	If one sets $H_{ab} = f g_{ab}$, for some strictly positive or strictly negative smooth function\footnote{Non-degeneracy of $H_{ab}$ demands $f\neq 0$. Hence, $f$ is either strictly positive or strictly negative.\vfill} $f$, the generalized proper time action~\eqref{eq:ModifiedProperTime} is recovered. Moreover, the connection reduces to one of Weyl-type with an exact $1$-form $\omega_a = \partial_a \omega$, whose potential is given by $\omega = -\ln|f|$.
\end{remark}

\begin{remark}
	In the Symmetric Teleparallel Equivalent of General Relativity (STEGR), one postulates the curvature of the affine connection to vanish. Together with the torsion-freeness of the connection, this implies the existence of the so-called \emph{coincident gauge}. This is a choice of chart for which the Christoffel symbols of $\nabla$ vanish not just in a point, but in an entire open neighbourhood.
	
	In the present context, the effective metric $H_{ab}$ offers an alternative view on the coincident gauge. Since $\bar{R}_{abcd}$ is now postulated to vanish, and since $\nabla$ is torsion-free and compatible with $H_{ab}$ in the sense of equation~\eqref{eq:MetricCompatibility}, we are effectively dealing with a flat metric $H_{ab}$. In a generic gauge, the components of $H_{ab}$ are of course functions of the coordinates. Just like in a generic gauge the Minkowski metric is not a diagonal matrix with constant entries. However, due to the flatness condition, we are guaranteed that there exists a chart in which $H_{ab}$ is constant. Thus, the autoparallel equation is essentially describing geodesics of $H_{ab}$, which in the chart in which $H_{ab}$ is constant become straight lines described by $\ddot{x}^{a}=0$. 
\end{remark}

\begin{remark}
	In the absence of non-metricity, we recover the geodesic equation. This follows directly from~\eqref{eq:FinalEq}. Because of~\eqref{eq:ChristoffelSymbolsH}, $\Gamma\ud{a}{bc} = \left\{a\atop bc\right\}$, and because the Levi-Civita connection is unique, it immediately follows that $H_{ab} = g_{ab}$. In this particular special case, we are guaranteed that $H_{ab}$ is non-degenerate.
\end{remark}

\begin{remark}
	The only obstruction to the existence of an action functional and thus to the variational nature of the autoparallel equation is the degeneracy of $H_{ab}$. It is conceivable that not all solutions to~\eqref{eq:MetricCompatibility} are non-degenerate. 
\end{remark}

\section{Discussion}\label{sec:discussion}
In metric-affine geometries $(\M, g, \nabla)$, upon which the geometric trinity of gravity is based, geodesics and autoparallel curves are genuinely distinct concepts. While geodesics can be derived from an action functional, which has the interpretation of proper time, autoparallel curves could so far not be described in terms of a variational principle. Even though generalizations of proper time functionals in the context of torsion-free manifolds with non-vanishing non-metricity have been studied in the literature~\cite{Romero_2012,Avalos17, DelhomTime}, these functionals can at best produce the autoparallel equation of a Weyl-type connection.

In this letter we proved that an action functional exists, that it is given by
\begin{align*}
	\boxed{\mathcal{S}[\gamma] = \int_I \sqrt{|H_{ab}\dot{\gamma}^{a} \dot{\gamma}^b|}\,\dd\lambda}\nonumber
\end{align*}
and that the autoparallel equation of a torsion-free, but not necessarily metric-compatible connection follow from it via a standard variational procedure:
\begin{align*}
		\boxed{\ddot{x}^{a} + \left(\left\{d \atop bc\right\} + L\ud{d}{bc}\right)\dot{\gamma}^{b} \dot{\gamma}^c = 0} 
\end{align*}
We arrived at this result not by generalizing the functionals described in~\cite{Romero_2012,Avalos17, DelhomTime} or by otherwise guessing the correct form. Instead, we started from the autoparallel equation itself and solved the inverse problem of the calculus of variation.

We found that it is necessary to introduce a symmetric, non-degenerate tensor $H_{ab}$, which satisfies the condition~\eqref{eq:MetricCompatibility}, in order to solve the problem (or similarly~\eqref{eq:ChristoffelSymbolsH}). Because of these properties, we can think of $H_{ab}$ as an effective metric which coexists with the metric $g_{ab}$. Moreover, because of~\eqref{eq:MetricCompatibility} this effective metric encodes information about non-metricity and the Christoffel symbols of the metric $g_{ab}$. Crucially, we proved that the condition~\eqref{eq:MetricCompatibility} can, in principle, always be solved for $H_{ab}$ for a prescribed $g_{ab}$ and a prescribed non-metricity tensor $Q_{abc}$. Hence, the functional~\eqref{eq:ActionQ} is well-defined and produces the autoparallel equations for vanishing torsion but non-vanishing non-metricity.

The only obstruction to the existence of an action functional for the autoparallel equation is the possibility of obtaining degenerate solutions to~\eqref{eq:MetricCompatibility}. Whether such solutions exist cannot be answered at the present stage. However, it is possible to shed some light on this issue by solving~\eqref{eq:MetricCompatibility} in simple models already within GR. We leave this to future work. An other interesting avenue to pursue is to investigate whether analogous results can be proved for metric-affine geometries with torsion and vanishing non-metricity.

We consider it highly interesting to investigate the physical consequences of the disformation force for the motion of freely falling particles in cosmological and astrophysical environments. In particular, the modified geodesic equations derived in equation~\eqref{eq:autoparallelQ} provide a natural framework for studying how non-metricity affects the trajectories of test bodies and the relative acceleration of nearby worldlines. Such effects can be systematically explored in highly symmetric backgrounds. For instance, in cosmological settings one may employ the metric and connection ans\"{a}tze compatible with homogeneous and isotropic geometries \cite{DAmbrosio:2021pnd}, while in astrophysical contexts one can consider the corresponding spherically symmetric configurations \cite{DAmbrosio:2021zpm} describing compact objects such as stars and black holes. Solving the modified geodesic equations in these backgrounds would reveal how the disformation force alters orbital motion, radial infall, and geodesic deviation. While the modified geodesic equations determine the motion of test particles, the gravitational field itself may be specified independently through the choice of the underlying metric-affine theory.
Solving the modified geodesic equations in these backgrounds would reveal how the disformation force modifies orbital motion, radial infall, and geodesic deviation. Concerning the gravitational dynamics, one may remain within STEGR, in which the gravitational field equations are derived from the non-metricity scalar $\mathbb{Q}$~\cite{BeltranJimenez:2017tkd}. The analysis can also be extended to modified theories of the $f(\mathbb{Q})$ gravity type, where the gravitational action is generalized to a function 
$f(\mathbb{Q})$. In either case, the equations governing particle motion follow from the variational principle derived above, while the background geometry is determined by the corresponding gravitational field equations.

The phenomenological implications of such modifications in the geodesic equation could be far-reaching, much like those arising from modifications of the gravitational dynamics~\cite{BeltranJimenez:2019tme, DAmbrosio:2021zpm}. In cosmology, deviations from standard geodesic motion may influence the propagation of matter in the late universe and modify the behavior of large-scale structure tracers. In particular, the disformation force could lead to corrections in the relative acceleration of neighboring geodesics, thereby affecting the growth of cosmic structures and potentially leaving imprints in observables related to galaxy clustering or peculiar velocities.

In the strong-gravity regime around compact objects, the modified equations of motion may lead to observable consequences in the dynamics of particles and photons orbiting black holes. For example, corrections to circular and quasi-circular orbits could modify the location of the innermost stable circular orbit (ISCO), which plays a crucial role in accretion disk physics and in the emission spectra of compact objects. Similarly, alterations in null geodesics may affect the structure of photon spheres and the associated lensing properties of black holes. This could in turn lead to measurable modifications of the black hole shadow, whose size and shape are currently being probed by very long baseline interferometry observations such as those performed by the Event Horizon Telescope. Even small deviations in the effective trajectories of light rays could translate into detectable distortions of the shadow boundary or changes in the lensing pattern near the photon sphere.
Furthermore, the modified geodesic deviation induced by non-metricity may affect tidal forces in the vicinity of compact objects and alter the dynamics of inspiralling bodies in strong gravitational fields. This could potentially lead to corrections in gravitational-wave signals emitted during extreme mass-ratio inspirals or binary mergers. Observatories such as LIGO Scientific Collaboration and Virgo Collaboration are already probing such strong-field regimes with increasing precision, making the study of non-metricity-induced effects particularly timely.

From a practical perspective, the analysis of these phenomena follows a straightforward procedure. One inserts the appropriate symmetry-compatible ansätze for the metric and the affine connection into the modified geodesic equations~\eqref{eq:autoparallelQ} and solves for the resulting particle trajectories. If a variational formulation is desired, the corresponding action~\eqref{eq:ActionQ} can be employed after computing the tensor 
$H_{ab}$ from the metric and the connection using equation~\eqref{eq:MetricCompatibility} (or similarly~\eqref{eq:ChristoffelSymbolsH}). A detailed investigation of these cosmological and astrophysical implications lies beyond the scope of the present work and will be pursued in future studies.


\section*{Acknowledgements}
LH would like to acknowledge financial support from the European Research Council (ERC) under the European Unions Horizon 2020 research and innovation programme grant agreement No 801781. LH further acknowledges support from the Deutsche Forschungsgemeinschaft (DFG, German Research Foundation) under Germany's Excellence Strategy EXC 2181/1 - 390900948 (the Heidelberg STRUCTURES Excellence Cluster). 


\bibliographystyle{apsrev4-2}
\bibliography{reference}

\end{document}